%% file: paper.tex
\documentclass[twoside,leqno,twocolumn,english]{article}
\usepackage{ltexpprt}

\usepackage{array,booktabs,ragged2e,multirow}
\newcolumntype{R}[1]{>{\RaggedLeft\arraybackslash}p{#1}}
\newcolumntype{L}[1]{>{\RaggedRight\arraybackslash}p{#1}}
\newcolumntype{C}[1]{>{\centering\arraybackslash}p{#1}}
\usepackage[labelfont=bf]{caption}

\usepackage{caption}
\usepackage{xspace}
\usepackage{amsmath}
\usepackage{amsfonts}
\usepackage{amssymb}
\usepackage{pifont}
\usepackage[usenames,dvipsnames]{xcolor}
\usepackage{algorithm}
\usepackage{algorithmic}

\input{dfn.tex}

\usepackage{url}
\usepackage{paralist} 
\usepackage{setspace}
\usepackage{listings} 
\usepackage{hyperref}
\usepackage{xcolor}
\hypersetup{
	colorlinks,
	linkcolor={red},
	citecolor={black},
	urlcolor={blue}
}

\usepackage{tikz}
\usepackage{ctable}

\begin{document}

\title{\method~on Time-Evolving Graphs:\\ Succinctly Explaining Anomalies from Any Detector}

\author{
Nikhil Gupta\\
IIT Delhi \\
{\em cs5140462@cse.iitd.ac.in}\\ 
\and
Dhivya Eswaran, Neil Shah, Leman Akoglu, Christos Faloutsos\\
Carnegie Mellon University \\ 
{\em \{deswaran, neilshah, lakoglu, christos\}@cs.cmu.edu}\\
}

\date{\vspace{-5ex}}
\maketitle


\begin{abstract}

\input{000abstract}
\end{abstract}

\section{Introduction}
\label{sec:intro}
\input{010introduction}

\section{Preliminaries and Problem Statement}
\label{sec:prelim}

\input{015prelim.tex}


\input{020problem}

\section{Proposed Algorithm \method}
\label{sec:meth}

\input{030method}

\section{Experiments}
\label{sec:exp}
\input{040experiment}

\section{Related Work}
\label{sec:background}

\input{020background}

\section{Conclusions}
\label{sec:concl}
\input{060conclusion}

\bibliographystyle{abbrv}
\bibliography{BIB/other}

\appendix
\input{070proof}

\end{document}

%% file: dfn.tex
\newtheorem{problem}{Problem}

\newtheorem{definition}{Definition}
\newtheorem{case}{Case}

\newcommand{\hide}[1]{}

\newcommand{\reminder}[1]{{\textsf{\textcolor{red}{[#1]}}}}

\newcommand{\hadoop}{\textsc{Hadoop}\xspace}

\newcommand{\bit}{\begin{compactitem}}
\newcommand{\eit}{\end{compactitem}}
\newcommand{\ben}{\begin{compactenum}}
\newcommand{\een}{\end{compactenum}}

\newcommand{\ind}{{\texttt{indegree}}\xspace}
\newcommand{\outd}{{\texttt{outdegree}}\xspace}
\newcommand{\inwv}{{\texttt{inweight-v}}\xspace}
\newcommand{\outwv}{{\texttt{outweight-v}}\xspace}
\newcommand{\inwr}{{\texttt{inweight-r}}\xspace}
\newcommand{\outwr}{{\texttt{outweight-r}}\xspace}
\newcommand{\iatavg}{{\texttt{IAT-avg}}\xspace}
\newcommand{\iatvar}{{\texttt{IAT-variance}}\xspace}
\newcommand{\iatmin}{{\texttt{IAT-min}}\xspace}
\newcommand{\iatmax}{{\texttt{IAT-max}}\xspace}

\newcommand{\lifetime}{{\texttt{lifetime}}\xspace}

\newcommand{\mA}{{\mathcal{A}}}
\newcommand{\mVs}{{\mathcal{V}_s}}
\newcommand{\mVd}{{\mathcal{V}_d}}
\newcommand{\mE}{{\mathcal{E}}}
\newcommand{\mV}{{\mathcal{V}}}
\newcommand{\mF}{{\mathcal{F}}}
\newcommand{\mP}{{\mathcal{P}}}
\newcommand{\mS}{{\mathcal{S}}}

\newcommand{\method}{{\sc LookOut}\xspace}
\newcommand{\baseline}{{\sc LookOut-Na\"{i}ve}\xspace}

\newcommand{\tstrut}{\rule{0pt}{2.4ex}}
\newcommand{\bstrut}{\rule[-1.0ex]{0pt}{0pt}}

\newcommand{\tgraph}{{$t$-graph}\xspace}
\newcommand{\tgraphs}{$t$-graphs\xspace}

\newcommand{\eg}{\textit{e.g.,}\xspace}
\newcommand{\ie}{\textit{i.e.,}\xspace}
\newcommand{\graph}{\mathcal{G}\xspace}

\definecolor{OliveGreen}{rgb}{0,0.6,0}

\newcommand{\tick}{\textcolor{OliveGreen}{\ding{52}}}

\newcommand{\enron}{\textsc{Enron}\xspace}
\newcommand{\lbnl}{\textsc{LBNL}\xspace}
\newcommand{\dblp}{\textsc{DBLP}\xspace}
\newcommand{\explanationquality}{Quality of Explanation}
\newcommand{\scalability}{Scalability}
\newcommand{\discoveries}{Discoveries}
\newcommand{\mywidth}{1}

\newcommand{\incrimination}{incrimination\xspace}

%% file: 000abstract.tex
Why is a given node in a time-evolving graph ($t$-graph) marked as an anomaly by an off-the-shelf detection algorithm?  Is it because of the number of its outgoing or incoming edges, or their timings?  How can we best convince a human analyst that the node is anomalous?  Our work aims to provide succinct, interpretable, and simple explanations of anomalous behavior in $t$-graphs (communications, IP-IP interactions, etc.) while respecting the limited attention of human analysts.  Specifically, we extract key features from such graphs, and propose to output a few pair (scatter) plots from this feature space which ``best'' explain known anomalies. To this end, our work has four main contributions: (a) \textbf{problem formulation}: we introduce an ``analyst-friendly'' problem formulation for explaining anomalies via pair plots, (b) \textbf{explanation algorithm}: we propose a plot-selection objective and the \method algorithm to approximate it with optimality guarantees, (c) \textbf{generality}:  our explanation algorithm is \textit{both} {domain-} and {detector-agnostic}, and (d) \textbf{scalability}: we show that \method scales linearly on the number of edges of the input graph.  Our experiments show that \method performs near-ideally in terms of maximizing explanation objective on several real datasets including Enron e-mail and DBLP coauthorship.  Furthermore, \method produces fast, visually interpretable and intuitive results in explaining ``ground-truth'' anomalies from Enron, DBLP and LBNL (computer network) data.  

%% file: 010introduction.tex
Given a time-evolving graph (hereafter referred to as a $t$-graph) in the form $\langle$source, destination, timestamp, value$\rangle$ (such as IP-IP communications in bits or user-merchant spendings in dollars over time), and a list of anomalous nodes (identified by an off-the-shelf ``black-box'' detector or any other external mechanism), how can we  {\em explain} the anomalies to a human analyst in a succinct, effective, and interpretable fashion?

\begin{figure}[!t]	
	\begin{tabular}{c|c}
	\multirow{2}{*}[0.3\columnwidth]{\includegraphics[width=0.45\columnwidth]{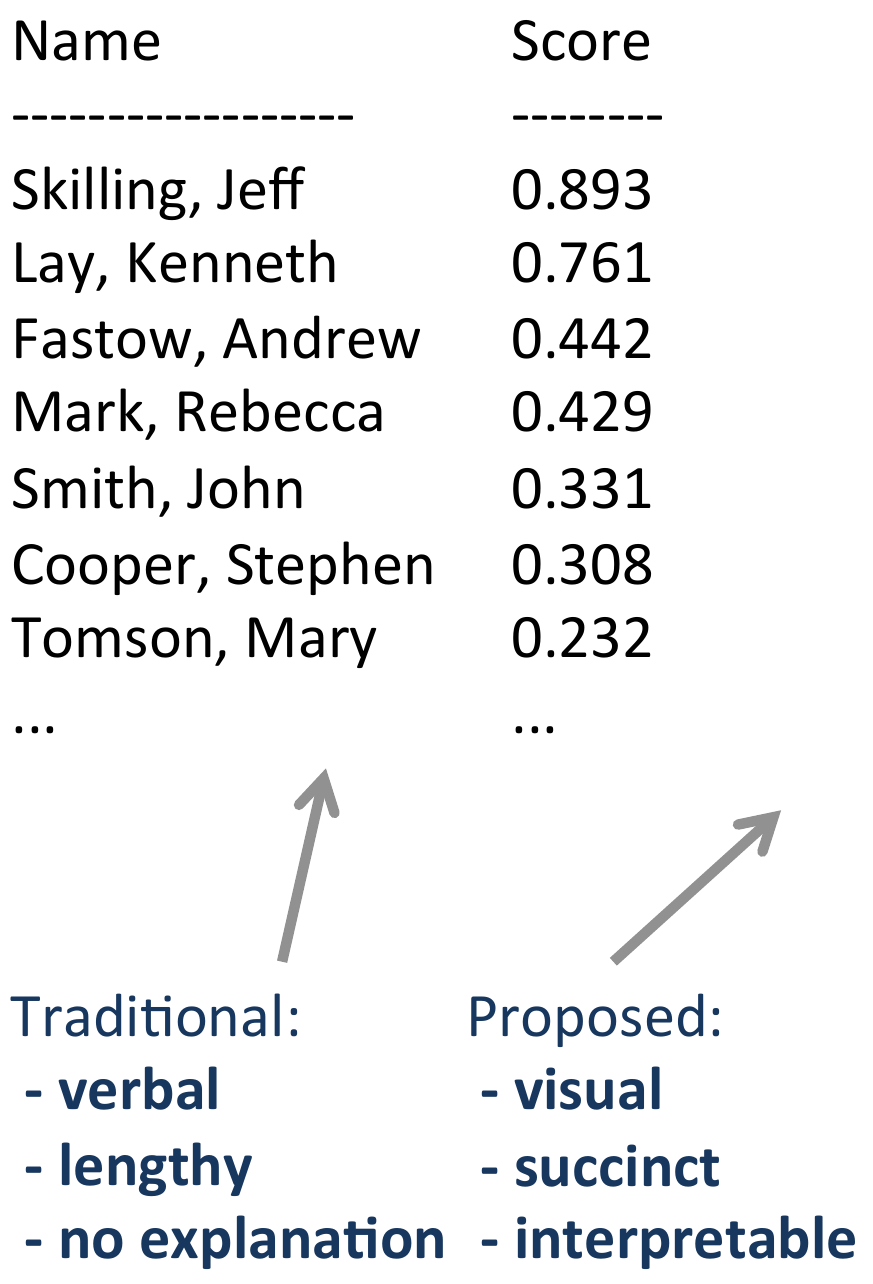}} & \includegraphics[width=0.5\columnwidth]{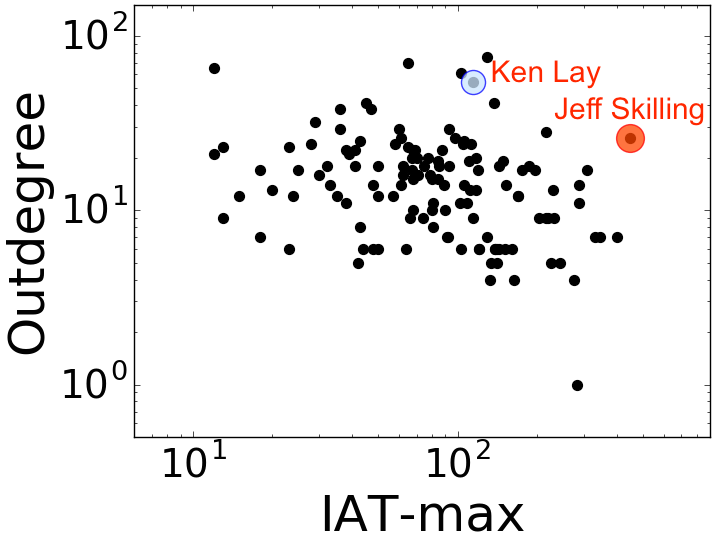} \\
	  &  \includegraphics[width=0.5\columnwidth]{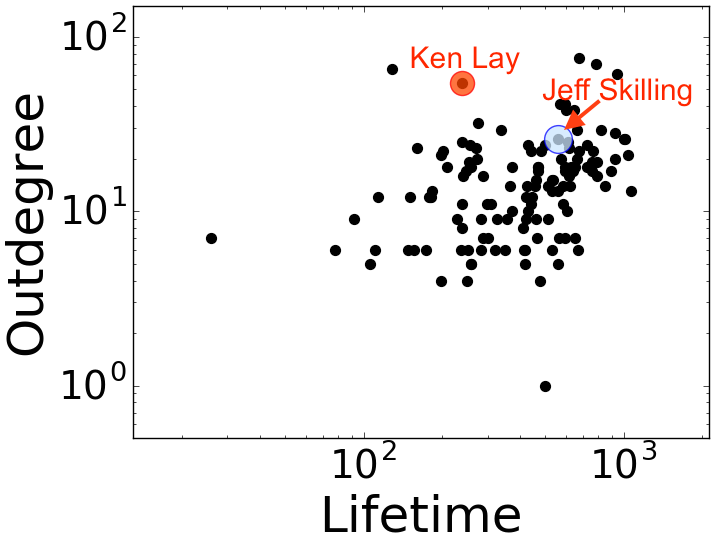}\\
	\end{tabular}
		\caption{\method explains Enron founder/CEO ``Ken Lay'' and COO ``Jeff Skilling'' by two ``pair plots'' in which they are most salient. 
			Compared to traditional 
			ranked list output (left), \method 
			produces simpler, more interpretable explanations (right).
			\label{fig:jewel}}
\vspace{-0.25in}	
\end{figure}

Anomaly detection is a widely studied problem. Numerous detectors exist for point data \cite{books/sp/Aggarwal2013,LOF,iForest}, time series \cite{journals/tkde/GuptaGAH14}, as well as graphs \cite{Oddball,journals/datamine/AkogluTK15}. 
However, the literature on anomaly explanation or description is perhaps surprisingly sparse.
Given that the outcomes (i.e. alerts) of a detector often go through a ``vetting'' procedure by human analysts, 
it is extremely beneficial to provide explanations for such alerts which can empower analysts in sensemaking and reduce their efforts in troubleshooting and recovery. Moreover, such explanations should justify the anomalies as succinctly as possible in order to save analyst time.

Our work sets out to address the anomaly explanation problem, with a focus on time-evolving graphs.  Consider the following
{\bf example situation:}
Given the IP-IP communications over time from an institution, 
a network analyst could face two relevant scenarios.
\bit 
\item {\em Detected anomalies:} For monitoring, s/he could use any ``black-box'' anomaly detector to spot suspicious IPs. Here, we are oblivious to the specific detector, knowing only that it flags anomalies, but does not produce any interpretable explanation.
\item {\em Dictated anomalies:} Alternatively, anomalous IPs may get reported to the analyst externally (e.g., they crash or get compromised).
\eit
In both scenarios, the analyst would be interested in understanding {\em in what ways} the pre-identified anomalous IPs (detected or dictated) differ from the rest.

In this work, we propose a new approach called \method, for explaining a given set of anomalies, and apply it to the scope of $t$-graphs. 
At its heart, \method provides interpretable {\em visual} explanations through simple, easy-to-grasp plots, 
which ``incriminate'' the given anomalies the most.
We summarize our contributions as follows.

\bit
\item {\bf Anomaly Explanation Problem Formulation:}
We introduce a new formulation that explains anomalies through ``pair plots''.
In a nutshell, given the list of anomalies from a $t$-graph, we aim to find a few pair plots on which the total ``blame'' that the anomalies receive is maximized.  Our emphasis is on two key aspects: 
(a) \textit{interpretability}: our plots visually incriminate the anomalies, and (b) \textit{succinctness}: we show only a few plots to respect the analyst's attention; the analysts can then quickly interpret the plots, spot the outliers, and verify their abnormality given the discovered feature pairs. 


\item {\bf Succinct Explanation Algorithm \method:} 
We propose the \method algorithm to solve our explanation problem.
Specifically, we develop a plot selection objective that lends itself to monotone submodular function optimization, which we solve efficiently with optimality guarantees.
Figure \ref{fig:jewel} illustrates \method's performance on the Enron communications network, where it discovers two pair plots which maximally incriminate the given anomalous nodes: Enron founder ``Ken Lay'' and CEO ``Jeff Skilling.''	 Note that the anomalies stand out visually from the normal nodes.

\item {\bf Generality:} 
\method is general in two respects: it is (a) {\em domain-agnostic}, meaning it is suitable for real-world $t$-graphs from various domains that involve human actors, and
(2) {\em detector-agnostic}, meaning it can be employed to explain anomalies produced by any detector or identified through any other mechanism (e.g., crash reports, customer complaints, etc.)

\item {\bf Scalability:}
We show that \method requires time linear on the number of edges in the input $t$-graph, which is simply the cost incurred to compute feature values. In fact, the pair plot scoring and selection steps are carefully designed to add only a constant cost, which is independent of graph size (see Lemma \ref{theo:scale} and Fig. \ref{fig:scalability}).

\eit
\vspace{0.05in}

We experiment with real $t$-graph datasets from diverse domains including e-mail and IP-IP communications, as well as co-authorships over time, which demonstrate the effectiveness, interpretability, succinctness and generality of our approach. 

 {\bf Reproducibility:} \method is open-sourced at \url{https://github.com/NikhilGupta1997/Lookout}, and our datasets are publicly available (See Section \ref{ssec:data}).

\hide{
Given a large number of 2D scatter plots, what is the best way to choose a limited number of these plots to best explain the outlier nature of the anomalies on a given dataset?

Here we propose \method, a method which gives a score to each scatter plot based on how well it is able to explain the nature of each of the outliers. Based on the limitation on the number of plots that can be selected, the algorithm will only choose the best scored plots with minimum intersection between plots.



}


%% file: 015prelim.tex
\begin{table}[!t]
	\begin{center}
		\caption{Symbols and Definitions \label{tab:dfn}}
		\vspace{-0.1in}
		\setlength{\arrayrulewidth}{.1em}
		\begin{tabular}{r l}
			\hline  
			\bf{Symbol} & \bf{Definition} \tstrut \bstrut\\ 
			\hline 
			$\mathcal{G}$  & Input graph, $\mathcal{G} = (\mathcal{V},\mathcal{E})$, $|\mathcal{V}|=n, |\mathcal{E}|=m$ \\ 
			$\mathcal{A}$  &  Input anomalies, $|\mA|=k$ \\  
			$\mathcal{F}$  &  Set of node features, $|\mF|=d$ \\  
			$\mathcal{P}$  &  Set of pair plots, $|\mP|=d(d-1)/2=l$ \\ 
			${s}_{i,j}$  & Anomaly score of $a_i\in \mA$ in plot $p_j\in \mP$ \\ 		
			$\mathcal{S}$  & Subset of selected plots  \\ 
			$f(\mathcal{S})$  & Explanation objective function \\
			$\Delta_f(p \ |\ \mathcal{S})$  & Marginal gain of plot $p$ w.r.t $\mathcal{S}$\\
			$b$  & Budget, i.e., maximum cardinality of $\mathcal{S}$ \\   \hline
		\end{tabular} 
	\end{center}
	\vspace{-0.3in}
\end{table}

\subsection{Notation}
Formally, a \tgraph~is a collection of time-ordered edges $\graph(\mVs\cup\mVd,\mE) = \{e_1, e_2, \ldots\}$ where each edge $e \in \mE$ is a tuple $\langle v_s, v_d, ts, val\rangle$. Here, $v_s \in \mVs$ and $v_d \in \mVd$ are the source and destination nodes of $e$, $ts$ is the time stamp associated with $e$, and $val$ is a value on $e$.
For unipartite graphs (like IP-IP) $\mVs=\mVd$, and for bipartite graphs (like user-ATM) $\mVs\cap \mVd = \emptyset$.
Edge value could be categorical (edge type, \eg~withdrawal vs. deposit) or numerical (edge weight, \eg~dollar amount). 
We denote the number of nodes by $|\mVs\cup \mVd|=|\mV|=n$. We refer to the number of edges $|\mE|=m$ as the size of $\graph$.  Note that such graphs can be considered multigraphs, as multiple edges between $v_s$ and $v_d$ are permitted.

The set of anomalous nodes given as input is denoted by
$\mA\subset \mV$, $|\mA| = k$. For bipartite graphs, anomalies could be suspicious sources (\eg~users) $\mA \subset \mVs$ or destinations (\eg~ATMs) $\mA \subset \mVd$.

To characterize the anomalies in $t$-graphs, we use a list of features for each node, denoted by 
$\mF = \{f_1, f_2, \ldots, f_d\}$. We describe our node features next.

\vspace{-0.05in}
\subsection{Features for $t$-graphs}
\label{ssec:features}

We characterize the nodes in a time-evolving graph with numerical properties they exhibit, and explain node anomalies using these derived features.
The literature is rich on feature extraction from graphs \cite{Oddball,conf/kdd/HendersonGLAETF11,conf/pakdd/GiatsoglouCSBFV15,conf/icdm/ShahBHAGMKF16,perozzi2014deepwalk,conf/kdd/GroverL16}.

Oddball \cite{Oddball} introduced and successfully used various relational features for anomaly detection in static weighted graphs. We leverage several of these features including (1) \ind, (2) \outd, (3) \inwv, and (4) \outwv, where weight of an edge $\langle v_s, v_d, ts, val\rangle$ is the value $val$. 
Since we work with temporal graphs, we also use (5) \inwr~and (6) \outwr, where edge weight depicts the number of repetitions of the edge (multi-edge).

In addition to structure, we have the time information. To this end, we leverage the inter-arrival times (IAT) between the edges. IAT distributions have been successfully used in a number of prior works on fraud detection \cite{conf/pakdd/GiatsoglouCSBFV15,conf/icdm/ShahBHAGMKF16}. We specifically use:
{\tt IAT} (7) {\tt average}, (8) {\tt variance}, (9-11) {\tt minimum}, {\tt median},  and {\tt maximum}.
We also define (12) \lifetime, the gap between the first and the last edge.

Admittedly, there are numerous other features one could extract from graphs.
In this work,  our focus is on the explanation problem, and thus we mainly build on existing features that were shown to be useful in prior anomaly and fraud detection tasks.
One could easily extend our list with others, such as recursive structural  features \cite{conf/kdd/HendersonGLAETF11} or node embeddings \cite{perozzi2014deepwalk,conf/kdd/GroverL16},
but we recommend using interpretable features that have direct meanings for analyst benefit, as well as scalable to compute for large graphs.

\hide{
A natural way to perform anomaly detection is to first extract real-valued features for the entities of interest (sources) and then apply one of the available state-of-the-art anomaly detection approaches. However, what features for sources does one extract from \tgraphs? We divide the space of possible features into two: \textit{(i) relational} features, which account for the graph structure and \textit{(ii) temporal} features which focus on the time aspect. Our empirical studies revealed promise in the following \textit{domain-agnostic} features:

\begin{itemize}
	\item \textbf{Relational features:}
	\begin{enumerate}
		\item Number of unique sources\footnote{\label{note1}Only if the set of \textit{source} elements is equivalent to the set of \textit{destination} elements}
		\item Number of edges incident
		\item Number of unique destinations
		\item Number of outgoing edges
		\item Total weight of outgoing edges
	\end{enumerate}
	\item \textbf{Temporal features:}
	\begin{enumerate}
		\item Lifetime
		\item Quantiles of inter-arrival times
		\item Variance in inter-arrival times
		\item Mean of inter-arrival times
		\item Median of inter-arrival times
	\end{enumerate}
\end{itemize}

}

%% file: 020problem.tex


\hide{
\subsection{Overview}

Given $\graph$ and $\mA$, we follow three main steps toward explanation.

{\em 1. Extracting features:} We extract features for all the nodes in $\graph$ (Section \ref{ssec:features}).
Simply put, those reflect the relational and temporal properties of the nodes.
We look to explain the anomalies within this feature space.

{\em 2. Scoring the anomalies:} 
The explanations should be simple and interpretable. Even better, they should be easy to illustrate to humans.
To this end, we use feature pairs. That is, we generate $d \choose 2$ 2-d spaces by creating all pairwise combinations.
We then score the nodes in $\mA$ within each 2-d space by anomalousness. In other words, each input anomaly $a_i \in \mA$ will receive an anomalousness score from each feature pair $p_j$, denoted as $s_{i,j}$ (Section \ref{ssec:scoring}).

These 2-d spaces are easy to illustrate visually in a pair plot:
\begin{definition}[Pair Plot] Given a pair of features $f_x, f_y \in \mF$,
a pair plot $p$ is exactly a scatter plot of all nodes $v\in \mV$ with corresponding $x$ and $y$ values.
Anomalies ``explained away'' by $p$ are denoted by $\mA_p \subseteq \mA$,
and marked in `red' on the visualization.
\end{definition}
Moreover, anomalies in 2-d are easy to interpret: \eg~``point $a$ has too many/too little $y=$dollars for its $x$=number of accounts''. Given the pair plot, the analyst can spot the anomalies simply by eyeballing and come up with such explanations without supervision.

{\em 3. Explaining the anomalies:} 
Let us denote the set of all $d(d-1)/2=l$ pair plots by $\mP$.
Even for small $d$, showing all the pair plots would be too overwhelming for the analyst.
Moreover, some anomalies could show up in multiple plots, causing redundancy.
Ideal is to identify only a small number of pair plots, which can ``blame'' or ``explain away'' the anomalies to the largest possible extent. In other words, our goal is to select a subset $\mS$ of $\mP$, on which every anomaly in  $\mA$ receives as high anomalousness score as possible (Section \ref{ssec:selecting}).

}

\vspace{-0.1in}
\subsection{Intuition \& Proposed Problem}

The explanations we seek to generate should be simple and interpretable. Moreover, they should be easy to illustrate to humans who will ultimately leverage the explanations.
To this end, we decide to use ``pair plots'' (=scatter plots) for anomaly justification, due to their visual appeal and interpretability. That is, we consider ${d \choose 2} = \frac{d(d-1)}{2}$ 2-d spaces by generating all pairwise feature combinations.  Within each 2-d space, we then score the nodes in $\mA$ by their anomalousness (Section \ref{ssec:scoring}).

Let us denote the set of all $d \choose 2$ pair plots by $\mP$.
Even for small values of $d$, showing all the pair plots would be too overwhelming for the analyst.
Moreover, some anomalies could redundantly show up in multiple plots.  Ideally, we would identify only a few pair plots, which could ``blame'' or ``explain away'' the anomalies to the largest possible extent. In other words, our goal would be to output a small subset $\mS$ of $\mP$, on which nodes in $\mA$ receive high anomaly scores (Section \ref{ssec:selecting}).

Given this intuition, we formulate our problem:
\begin{problem}[Anomaly Explanation]
\bit
	\item \textbf{Given} 
	\bit
	 \item[(a)] a list of anomalous nodes $\mA\subset \mV$ from a \tgraph, either (1) detected by an off-the-shelf detector or (2) dictated by external information, and
		\item[(b)] a fixed budget of $b$ pair plots;
	\eit
	\item \textbf{find} the best such pair plots $\mS\subset \mP$, $|\mS|=b$
	
	\item \textbf{to maximize} the total maximum anomaly score of anomalies that we can ``blame'' through the $b$ plots.
	\eit
\end{problem}

%% file: 030method.tex
\hide{
$T$-graphs occur in a number of scenarios such as e-commerce websites, communication and transportation networks. We give two examples below:
\begin{enumerate}
	\item \textbf{Computer networks} consists of multiple devices, each identified by a unique IP address, exchanging messages with each other. Thus, each message constitutes a temporal edge. In this case, source and destination sets are equal, \ie $\mathbb{S} = \mathbb{D}$ and consists of all IP addresses in the network. Edge metadata could include the number of packets  exchanged in the message (a non-negative integer) and the protocol used (categorical).
	\item \textbf{Credit card transactions data} consists of purchases made by consumers, each identified by a card or account number, at various merchants. Observe that the set of sources and destinations are disjoint in this example. Each edge (purchase) may be accompanied by an amount and location of transaction.
\end{enumerate}

Our work addresses the problem of anomaly detection and description in such \tgraphs. We formally state these below.
\begin{problem}[Anomaly Detection]
	\textbf{Given} a \tgraph $\graph = \{e_1, e_2, \ldots\}$ and a required number of anomalies $k$,
	\textbf{find} the $k$ most anomalous sources\footnote{with no loss of generality; to identify anomalous destinations, we need only swap the order of first two entries in each temporal edge tuple.}.
\end{problem}

}

In this section, we detail our approach for scoring the input anomalies by pair plots, our selection criterion and algorithm for choosing pair plots, the overall complexity analysis of \method, and conclude with discussion.

\vspace{-0.05in}
\subsection{Scoring by Feature Pairs}
\label{ssec:scoring}

Given all the nodes $\mV$, with marked anomalies $\mA\subset \mV$, and their extracted features $\mF \in \mathbb{R}^d$, our first step is to quantify how much ``blame'' we can attribute to each input anomaly in $\mathbb{R}^2$.
As previously mentioned, 2-d spaces are easy to illustrate visually with pair plots. Moreover, anomalies in 2-d are easy to interpret: \eg~``point $a$ has too many/too little $y$=dollars for its $x$=number of accounts''. Given a pair plot, an analyst can easily discern the anomalies visually, and come up with such explanations without any further supervision.

We construct 2-d spaces $(f_x,f_y)$ by pairing the features $\forall x,y=\{1,\ldots,d\}, x\neq y$. Each pair plot $p_j\in \mP$ corresponds to such a pair of features, $j=\{1,\ldots,{d\choose{2}}\}$.
For scoring, we consider two different scenarios, depending on how the input anomalies were obtained.

If the anomalies are {\em detected} by some ``black-box'' detector available to the analyst, 
we can employ the same detector on all the nodes (this time in 2-d) and thus obtain the scores for the nodes in $\mA$. 

If the anomalies are {\em dictated}, i.e. reported externally, then the analyst could use any off-the-shelf detector, such as LOF \cite{LOF}, DB-outlier \cite{conf/vldb/KnorrN99}, etc.
In this work, we use the Isolation Forest (iForest) detector \cite{iForest} for two main reasons: (a) it boasts {\em constant} training time and space complexity due to its sampling strategy, and (b) it has been shown empirically to outperform alternatives \cite{emmott2013systematic} and is thus state-of-the-art.
However, note that none of these existing detectors has the ability to \textit{explain} the outliers, especially iForest, as it is an ensemble approach.

By the end of the scoring process, each anomaly receives a total of $|\mP|=l$ scores. 


\subsection{Explaining by Pair Plots}
\label{ssec:selecting}

While scoring in small, 2-d spaces is easy and can be trivially parallelized, 
presenting all such pair plots to the analyst would not be productive given their limited attention budget.
As such, our next step is to carefully select a short list of plots that best blame all the anomalies collectively, where the plot budget can be specified by the user.

While selecting plots for justification, we aim to incorporate the following criteria:
\bit
\item \textbf{incrimination power}; such that the anomalies are scored as highly as possible,
\item \textbf{high expressiveness}; where each plot incriminates multiple anomalies, so that the explanation is \textit{sublinear} in the number of anomalies, and
\item \textbf{low redundancy}; such that the plots do not repeatedly explain similar sets of anomalies.
\eit
\vspace{0.025in}
We next introduce our objective criterion which satisfies the above requirements.

\subsubsection{Objective function}
At this step of the process, we can conceptually think of a complete, weighted bipartite graph between the $k$ input anomalies $\mathcal{A} = \lbrace{a_1,\dotsc,a_k}\rbrace$ and $l$ pair plots $\mathcal{P} = \lbrace{p_1,\dotsc,p_l}\rbrace$, in which edge weight $s_{i,j}$ depicts the
anomaly score that $a_i$ received from $p_j$, as illustrated in Fig. \ref{toy}.

\begin{figure}[tb!]
	\centering
	\begin{tabular}{cc}
		\includegraphics[width=0.8\columnwidth]{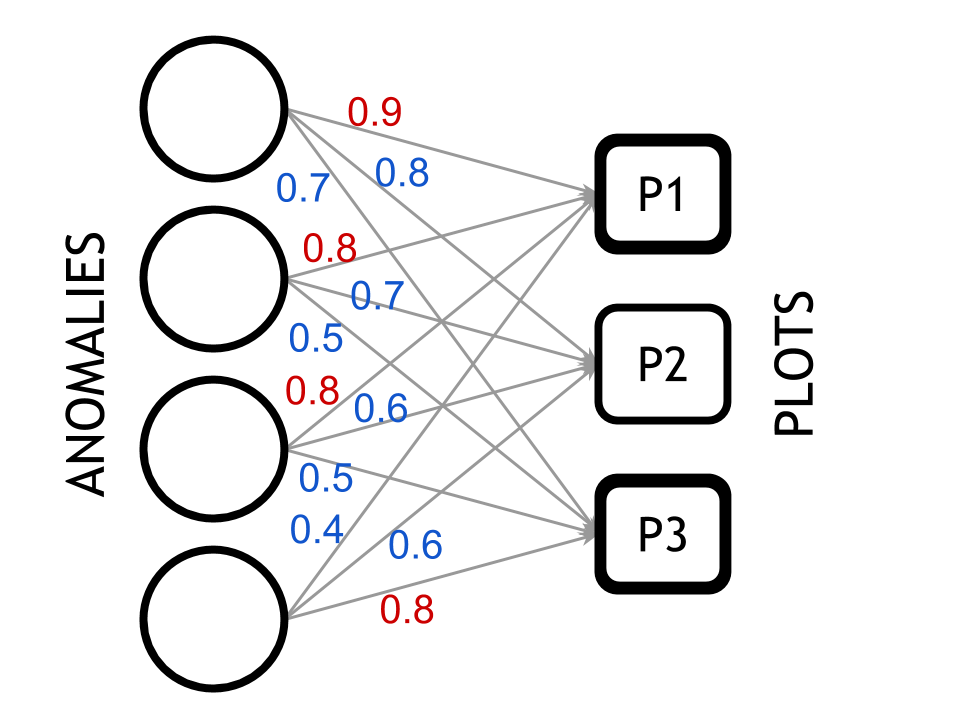}
	\end{tabular}
\vspace{-0.1in}
	\caption{\method~selection with $k=4$ anomalies, $l=3$ plots, and budget $b=2$. P1 is picked first due to maximum \emph{total incrimination} (2.9) (sum of incoming edge weights). Next P3 is chosen over P2, due to higher marginal gain (0.4 vs 0.2). 		\label{toy}}
\end{figure}

We formulate our objective to \textit{maximize} the total \textit{maximum} score of each anomaly amongst the selected plots, as given in Eq. \eqref{obj}:
\begin{equation}
\label{obj}
{\max_{{\mS\subset \mP, |\mS|=b}}}\;\;\; f(\mathcal{S}) = \sideset{}{\limits_{a_i \in \mathcal{A}}}\sum \;\;\max_{p_j \in \mathcal{S}}\;s_{i,j}  
\end{equation}
Here, $f(\mS)$ can be considered the \emph{total incrimination} score given by subset $\mS$.
Since we are limited with a budget of plots, we aim to select those which explain multiple anomalies to the best extent.  Note that each anomaly receives their maximum score from exactly one of the plots among the selected set, which effectively \textit{partitions} the explanations and avoids redundancy. 
In the example Fig. \ref{toy}, P1 and P3 ``explain away'' anomalies $\{1,2,3\}$ and $\{4\}$ respectively, where the maximum score that each anomaly receives is highlighted in red font.

Concretely, we denote by $\mA_p$ the set of anomalies that receive their highest score from $p$, i.e.
$\mA_p = \{a_i | p = \max_{p_j \in \mathcal{S}}\;s_{i,j}\}$, where we break ties at random.
Note that $\mA_p\cap \mA_{p'} = \emptyset, \; \forall p,p'\in \mP$.
In depicting a plot $p$ to the analyst, we mark all anomalies in $\mA_p$ in red and
the rest in $\mA\backslash \mA_p$ in blue -- c.f. Fig. \ref{fig:jewel}.

\subsubsection{Subset selection}

\begin{algorithm}[!t]
	\caption{{\sc \method}\label{alg:main}} 
	\begin{algorithmic}[1]
		\REQUIRE  anomalies $\mathcal{A}$, plots $\mP$, budget $b$
		\ENSURE  set of selected plots $\mathcal{S}$
		\STATE initialize $\mathcal{S} \to \emptyset$
		\STATE calculate $\Delta_f(p \ |\ \mathcal{S}) \ \forall \ p \in \mathcal{P}$
		\STATE generate sorted queue $Q$ on marginal gains
		\WHILE{$|\mathcal{S}| < b$}
		\STATE $p_{top}:=$ pop from $Q$
		\STATE update $\Delta_f(p_{top} \ |\ \mathcal{S})$
		\IF{$\Delta_f(p_{top} \ |\ \mathcal{S}) \geq$ top of $Q$  {\em \{ranks top\}}}
		\STATE $\mathcal{S} := \mathcal{S} \cup \lbrace p_{top} \rbrace$
		\ELSE
		\STATE insert $p_{top}$ in $Q$ with updated $\Delta_f(p_{top} \ |\ \mathcal{S})$
		\ENDIF
		\ENDWHILE
		\RETURN $\mathcal{S}$ 
	\end{algorithmic}
\end{algorithm}  
\setlength{\textfloatsep}{0.1in}

Having defined our explanation objective, we need to devise a subset selection algorithm to optimize Eq. \eqref{obj}, for a budget $b$. Note that optimal subset selection is a combinatorial task.

Fortunately, our objective $f(\cdot)$ exhibits three key properties that enable us to use a greedy  algorithm with an approximation guarantee. Specifically, we can show that our set function  $f:2^{\mathcal{P}}\rightarrow \mathbb{R}_{+}$ is 
\ben
	\item[i.]  \textit{non-negative}; since the anomaly scores take non-negative values, often in $[0,1]$;
	\item[ii.]  \textit{monotonic};  as for every $\mathcal{S} \subseteq \mathcal{T}\subseteq \mathcal{P}$, $f(\mathcal{S}) \leq f(\mathcal{T})$. That is, adding more plots to a set cannot decrease the maximum score attributed to any anomaly;
	\item[iii.]  and \textit{submodular}; since for every $\mathcal{S} \subseteq \mathcal{T} \subseteq \mathcal{P}$ and for all $p \in \mathcal{P}\backslash \mathcal{T}$, 
	$
	f(\mathcal{S}\cup \{p\}) - f(\mathcal{S}) \geq f(\mathcal{T}\cup \{p\}) - f(\mathcal{T})\;.
	$
	That is, adding any plot $p$ to a smaller set can increase the function value \emph{at least} as much as adding it to its superset. 
\een
We give the submodularity proof in Appendix \ref{proof}, as non-negativity and monotonicity are trivial to show.

Submodular functions with non-negativity and monotonicity properties admit approximation guarantees under a greedy approach identified by Nemhauser et al. \cite{journals/mor/NemhauserW78}. The greedy algorithm starts with the empty set $\mS_0$. In iteration $t$, it adds the element (in our case, pair plot) that maximizes the {\em marginal gain} $\Delta_f$ in function value, defined as
\vspace{-0.1in}
$$\Delta_f(p|\mS_{t-1}) = f(\mS_{t-1}\cup \{p\}) - f(\mS_{t-1})\;.$$
\vspace{-0.1in}
That is,
$$
\mS_t := \mS_{t-1} \cup \{\arg \max_{p\in \mathcal{P}\backslash \mS_{t-1}}  \Delta_f(p|\mS_{t-1}) \}\;.
$$
\vspace{-0.1in}
Based on their proof, one can show that when $t=b$,

\vspace{-0.1in}
$$f(\mS_b) \geq \big(1-\frac{1}{e}\big) \max_{|\mS|\leq b} f(\mS) \;.$$
\vspace{-0.1in}

In other words, the simple greedy search heuristic achieves at least $63$\% of the objective value of the \textit{optimum} set.

Algorithm~\ref{alg:main} shows the psuedocode of our \method pair plot selection process.


\hide{

\begin{definition}[marginal gain]
\label{def:gain}
Consider a plot $p \in \mathcal{P} \setminus \mathcal{S}$. The the marginal gain of $p$ with respect to $\mathcal{S}$ is defined as: $$\Delta_f(p \ | \ \mathcal{S}) = f(\mathcal{S} \cup \lbrace p \rbrace) - f(\mathcal{S})$$ 
\end{definition}

Marginal gain quantitatively measures the impact a new plot $p$, added to our chosen set of plots, will have in increasing the summarization metric. It helps to take in consideration any overlapping descriptions of the data between plots and only accounts for the added contributions of the new plot w.r.t the existing set.

A very important observation about the summarization metric is that it is submodular in nature\cite{Submodular}. 
\begin{definition}
SUBMODULARITY:
A function $f : 2^\mathcal{P} \to {\rm I\!R}$ is \textit{submodular} if for every $\mathcal{A} \subseteq \mathcal{B} \subseteq \mathcal{P}$ and $p \in \mathcal{P} \setminus \mathcal{B}$ it holds that $$\Delta_f(p \ | \ \mathcal{A}) \geq \Delta_f(p \ | \ \mathcal{B})$$
\end{definition}

The intuition behind submodularity is that the benifit gained by adding a new element in a set will decrease as more and more elements are made part of the set. This is a direct implication of overlapping gains between elements. A proof that the metric follows submodularity is given in the appendix. We will exploit the properties of submodular functions while designing our algorithm in the later section.

\reminder{Visualization details:} what is gray what is red etc.

Our algorithm \method starts of with an empty set and incremently chooses the best plot greedily based on the values of the marginal gains of the remaining unchosen plots. It takes advantage of the submodularity propertly when calculating the marginal gains of all the unselected plots. As a result of the submodular nature of the summarization metric we compute the marginal gain of the top most plot in the sorted list. If this value is larger than the marginal gain of the next best plot from the previous itteration, then the current plot in examination will have a higher marginal gain than all the other remainaing plots in the current iteration. This leads to a lazy update greedy algorithm and improves the running time considerably as seen in (Section~\ref{sec:exp}).
}


 \subsection{Complexity Analysis}
 \begin{lemma}
 \label{theo:scale}
 \method~scales linearly in terms of the input size, i.e. number of edges in the input $t$-graph.
 \end{lemma}

 \begin{proof}
The complexity of feature extraction (Section \ref{ssec:features}) is the most computationally demanding part of \method, and requires time linear on the number of edges in the graph---since all features are defined on the incident edges per node.

We show that the other two parts---scoring the given anomalies (Section \ref{ssec:scoring}) and selecting pair plots to present (Section \ref{ssec:selecting})---add only a constant time, and are \textit{independent of the graph size}.

Let us denote by $d$ the number of features we extract. For scoring, we generate 2-d representations by pairing the features.
For each pair, we train an iForest model in 2-d. Following their recommended setup, we typically \underline{sample} $256$ points (i.e., nodes) from the data and train around $t=100$ extremely randomized trees, called isolation trees.
The complexity of training an iForest is $O(t\psi\log \psi)$, and scoring a set of $|\mA|=k$ anomalous points takes $O(tk\log \psi)$.
The total scoring complexity is thus $O(d^2t(k+\psi)\log \psi)$, and is independent of the total number of nodes or edges in the input graph.

At last, we employ the greedy selection algorithm, where we compute the marginal gain  of adding each plot to our select-set, which takes $O(d^2k)$ and pick the plot with the largest gain through a linear scan in $O(d^2)$. We repeat this process $b$ times until the budget is exhausted. The total selection complexity is thus $O(d^2kb)$, which is also independent of graph size.

Thus, having extracted the features in $O(m)$, \method takes constant additional time to produce its output, where all of $\{d,\psi,t,k,b\}$ are (small) constants in practice.
\hfill $\blacksquare$
 \end{proof}

\input{031discussion-christos.tex}

%% file: 031discussion-christos.tex
\subsection{Discussion}
Here we answer some questions that may be in the reader's mind.

\noindent \textit{1. How do we define ``anomaly?''} 
We defer this question to the off-the-shelf 
anomaly detection algorithm
(iForest \cite{iForest}, LOF \cite{LOF}, DB-outlier \cite{conf/vldb/KnorrN99}, etc.)
Our focus here is to succinctly and interpretably show what makes the pre-selected items
stand out from the rest.

\vspace{1mm}

\noindent \textit{2. Why pair plots?} 
Using pair plots for justification is an essential, concious choice we make for several reasons:
(a) scatterplots are easy to look at and quickly interpret 
(b) they are universal and non-verbal, in that we need not use language to convey the anomalousness of points -- even people unfamiliar with the context of Enron will agree that the point ``Jeff Skilling'' in Fig. \ref{fig:jewel} is far away from the rest, and
(c) they show where the anomalies lie {\em relative to the normal points} -- the contrastive visualization of points is {more convincing} than stand-alone rules.

\vspace{1mm}

\noindent \textit{3. How do we choose the budget?} We designed our explanation objective to be 
budget-conscious, 
and let the budget be specified by the analyst. In general, humans have a working memory of size ``seven, plus or minus two'' \cite{miller1956mns}.


%% file: 040experiment.tex
In this section, we empirically evaluate \method on three, diverse datasets. Our experiments were designed to answer the following questions:
\begin{compactenum}
	\item[\textbf{[Q1]}] \textbf{\explanationquality:} How well can \method ``explain'' or ``blame'' the given anomalies?
	\item[\textbf{[Q2]}] \textbf{\scalability: }How does \method scale with the input graph size and the number of anomalies?
	\item[\textbf{[Q3]}] \textbf{\discoveries: }Does \method lead to interesting and intuitive explanations on real world data?
\end{compactenum}
These are addressed in Sec.~\ref{ssec:quality}, \ref{ssec:scalability} and \ref{ssec:discoveries} respectively. Before detailing our empirical findings, we describe the datasets used and our experimental setup.

\subsection{Dataset Description}\label{ssec:data}
To illustrate the generality of our proposed domain-agnostic $t$-graph features and \method itself, we select our datasets from diverse domains: e-mail communication (\enron), co-authorship (\dblp) and computer (\lbnl) networks. These datsets are publicly available, unipartite directed graphs. We describe them below and provide a summary in Table~\ref{tab:datasets}.

\renewcommand{\mywidth}{0.37in}
\begin{table}[tb!]
\caption{Datasets studied.\label{tab:datasets}}
	\centering
	\begin{tabular}{C{0.7in}C{\mywidth}C{\mywidth}C{\mywidth}C{0.63in}}
		\toprule
		\textbf{Dataset} & \textbf{Nodes} & \textbf{Edges} & \textbf{Time} & \textbf{Descrip.}\\
		\midrule
		\enron \cite{shetty2004enron}&  151 & 20K & 3 yrs & {\scriptsize email comm.} \\
		\dblp \cite{dblp} & 1.3M & 19M & 25 yrs & {\scriptsize coathorship}\\
		\lbnl \cite{DBLP:conf/imc/PangABLPT05} & 2.5K & 0.2M & 1 hr & {\scriptsize IP-IP comm.}\\
		\bottomrule
	\end{tabular}
\end{table}

\noindent \enron: This dataset consists of 19K emails exchanged between 151 \enron employees during the period surrounding the scandal\footnote{\url{https://en.wikipedia.org/wiki/Enron_scandal}} (May 1999-June 2002). The communications are on daily granularity.

\noindent \lbnl: This dataset details network traffic between 2.5K computers during a one hour interval known to contain network attacks (\eg port scans). Traffic is recorded at second granularity.

\noindent \dblp: This dataset contains the co-authorship network of 1.3M authors over 25 years from 1990 to 2014.  The networks are collected at yearly granularity.

\subsection{Experimental Setup}
To obtain ``ground-truth'' anomalies for \method input, we use the iForest \cite{iForest} algorithm on features described in Sec.~\ref{ssec:features} for each $t$-graph. This yields a ranked list of points with scores in $[0,1]$ (higher value suggests higher abnormality), from which we pick the desired top $k$. Analogously, we use iForest for computing the anomaly score in each pair plot. We note  that the analyst is free to choose any anomaly detector(s) for both/either of the above stages, making \method \textit{detector-agnostic}. However, it is recommended that the same methods be used for both stages to ensure ranking similarities.

\begin{figure}[tb!]
	\centering
	\begin{tabular}{c|c}
		\includegraphics[width=0.47\columnwidth]{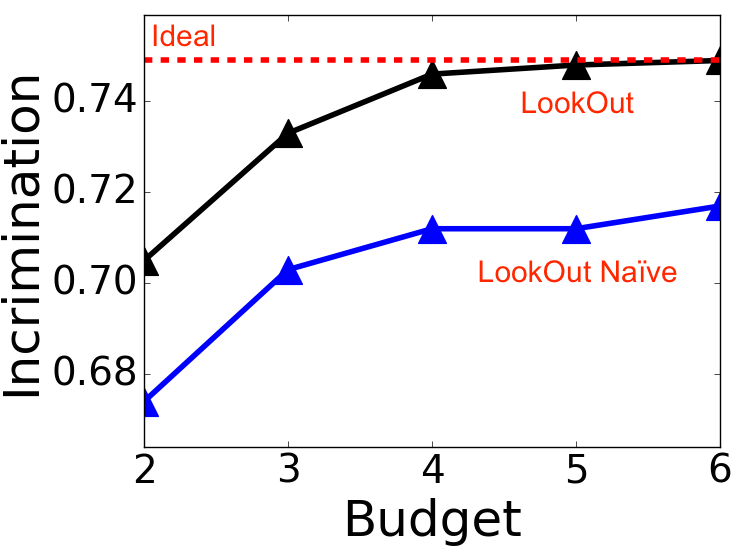} &
		\includegraphics[width=0.47\columnwidth]{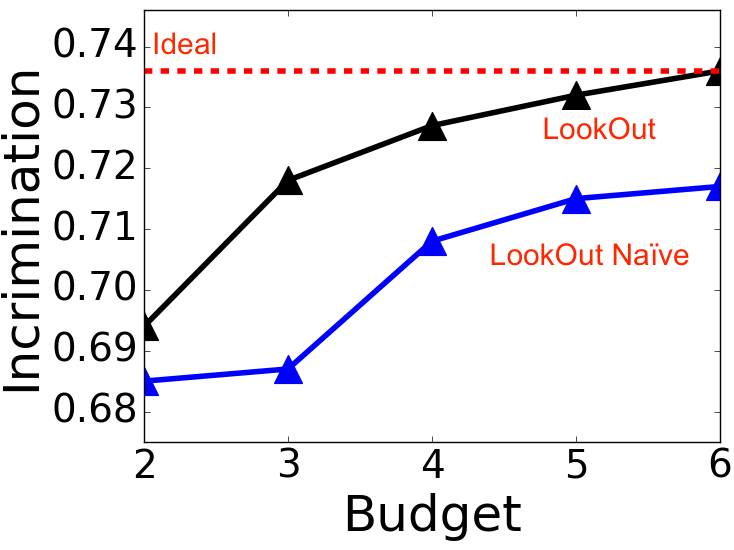}\\
		\includegraphics[width=0.47\columnwidth]{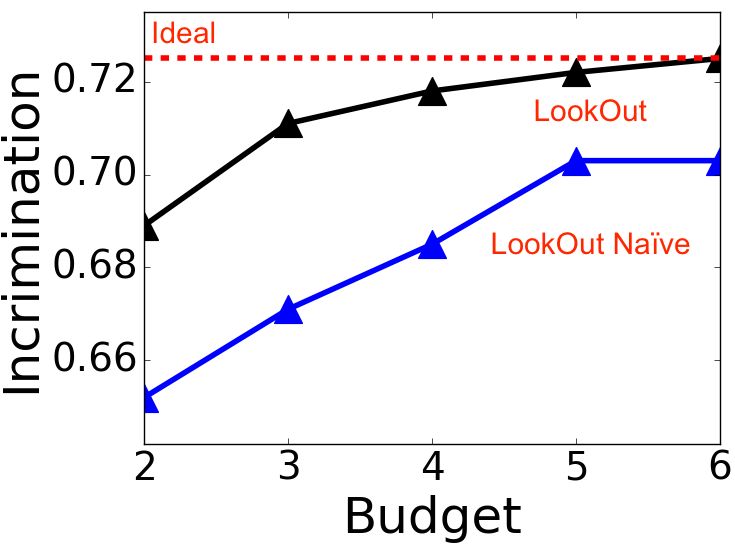} &
		\includegraphics[width=0.47\columnwidth]{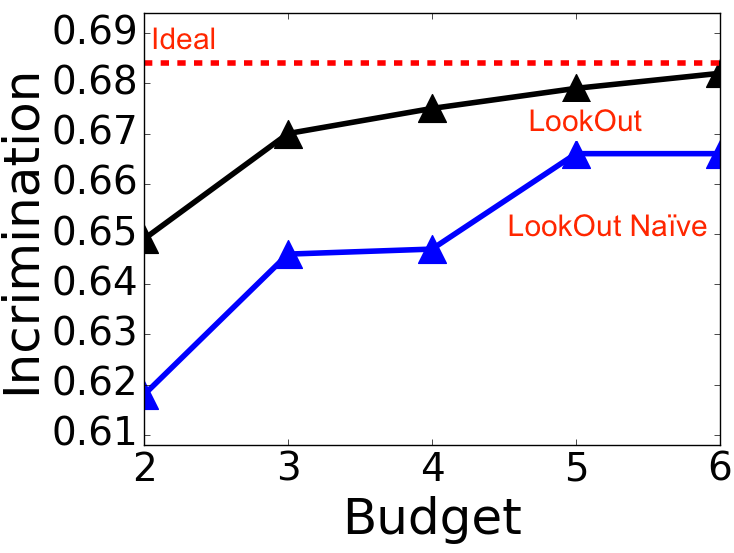}\\
		\includegraphics[width=0.47\columnwidth]{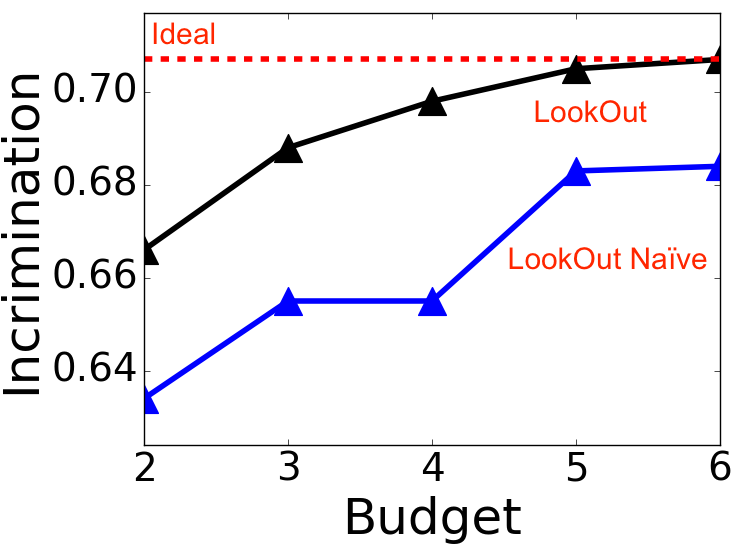} &
		\includegraphics[width=0.47\columnwidth]{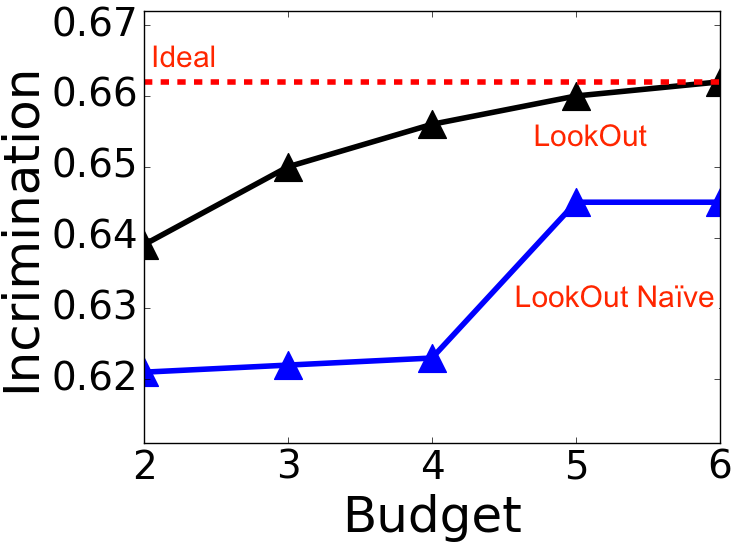}
	\end{tabular}
	\caption{\method explains away top 10-30 anomalies with 4-5 plots. Results shown on \lbnl (left) and \enron (right) with $k = \{10, 20, 30\}$ anomalies (top to bottom). \label{fig:quality}}
\end{figure}

\subsection{\explanationquality} \label{ssec:quality}
We quantify the quality of explanation provided by a set of plots $\mS$ through their \textit{\incrimination} score:
\begin{equation}
	\textit{incrimination}(\mS) =  \frac{f(\mathcal{S})}{\left| \mA \right|}
\end{equation}
where $\mA$ is the set of input anomalies and $f(\cdot)$ is the objective funtion defined in Eq.~\ref{obj}. Intuitively, \textit{\incrimination} is the average maximum blame that an input anomaly receives from any of the selected plots.

Due to the lack of comparable prior works, we use a na\"{i}ve version of our approach, called \baseline which ignores the submodularity of our objective. Instead, \baseline assigns a score to each plot by summing up scores for all given anomalies and chooses the top $b$ plots for a given budget $b$.

Fig.~\ref{fig:quality} compares the \textit{\incrimination} scores of both \baseline and \method on the \enron and \lbnl datasets for several choices of $k$ and $b$. The red dotted line indicates the ideal value, $f(\mP)$, \ie the highest achievable \incrimination (by selecting all plots). Fig.~\ref{fig:quality} shows that \method consistently outperforms \baseline and rapidly converges to the ideal \textit{incrimination} with increasing budget.  Results on \dblp were similar, but excluded for space constraints.

\subsection{\scalability} \label{ssec:scalability}
We empirically studied how \method runtime varies with (i) the graph size $m$ and (ii) the number of anomalies $k$. All experiments were performed on an OSX personal computer with 16GB memory. Runtimes were averaged over 10 trials.

To measure scalability with respect to graph size, we sampled edges from $[100K, 10M]$ in logarithmic intervals from \dblp while fixing budget to $b=5$ and $k=50$ anomalies. Fig.~\ref{fig:scalability} (left) plots runtime vs. number of edges, demonstrating a linear fit with slope $\approx$ 1. This empirically confirms Lem.~\ref{theo:scale}.

We also study the variation of runtime with the number of anomalies, as feature extraction incurs a constant overhead on each dataset.  Fig.~\ref{fig:scalability} (right) illustrates linear scaling with respect to number of anomalies for a \dblp subgraph with 10K edges.

\begin{figure}[tb!]
	\centering
	\begin{tabular}{cc}
		\includegraphics[width=0.48\columnwidth]{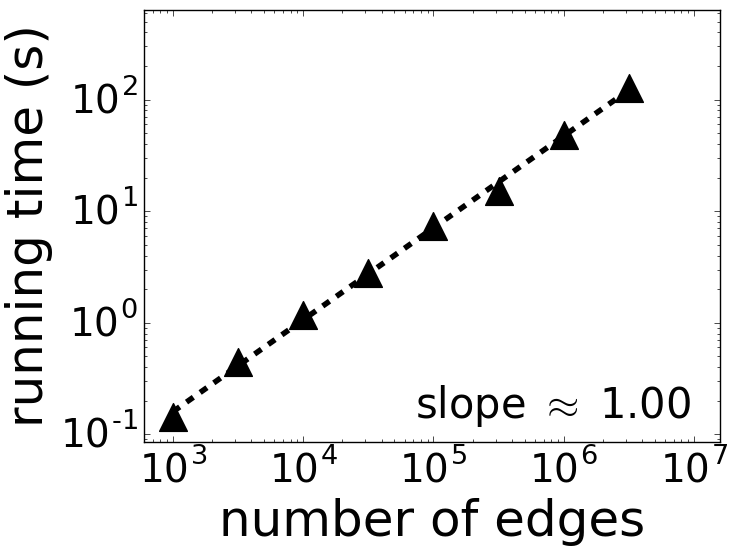} &
		\includegraphics[width=0.48\columnwidth]{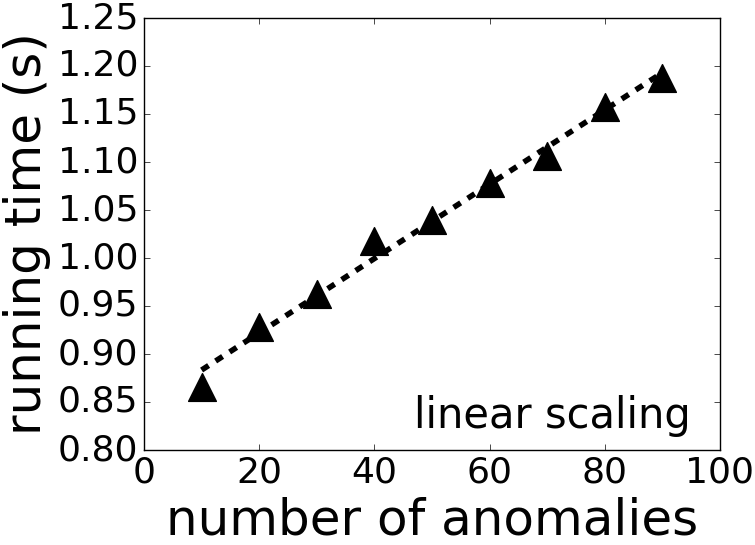}
	\end{tabular}
	\caption{\textbf{\method scales linearly} with number of edges (left) and number of anomalies (right). \label{fig:scalability}}
\end{figure}

\subsection{\discoveries} \label{ssec:discoveries}
\input{045discoveries}

%% file: 045discoveries.tex
\renewcommand{\mywidth}{0.46}
\begin{figure}[tb!]
	\centering
	\begin{tabular}{cc}
		\includegraphics[width=\mywidth\columnwidth]{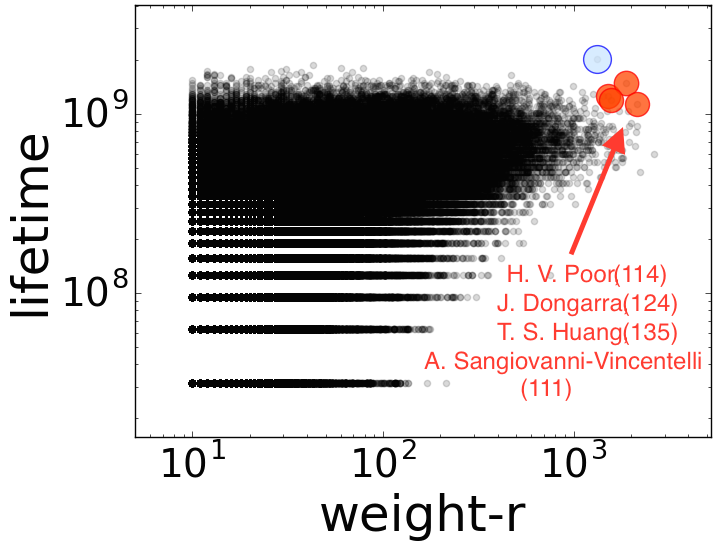}&
		\includegraphics[width=\mywidth\columnwidth]{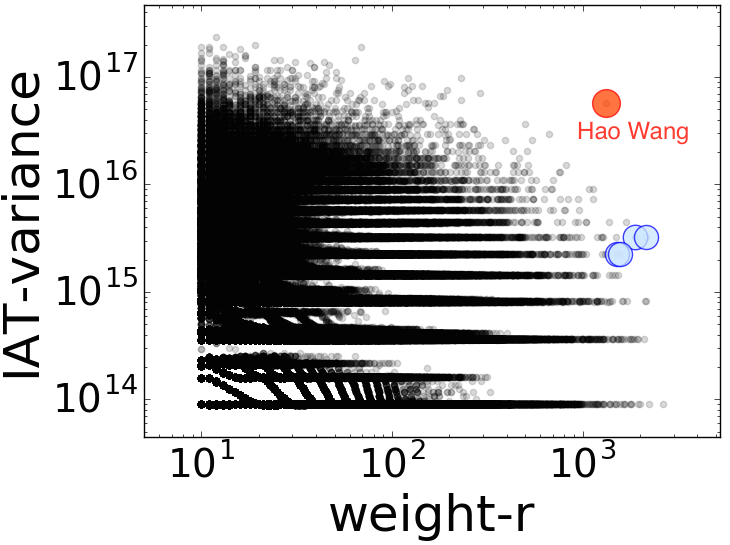}\\
		(a)& (b)\\[1mm]
		\hline\\
		\\
		\includegraphics[width=\mywidth\columnwidth]{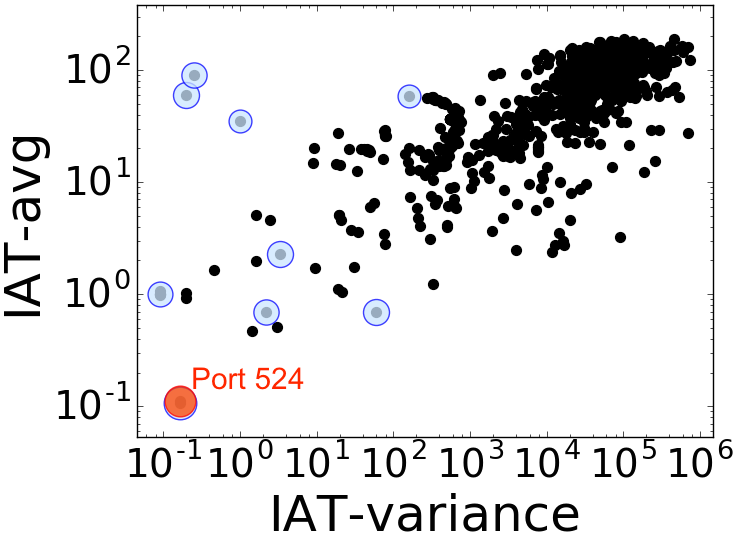} &
		\includegraphics[width=\mywidth\columnwidth]{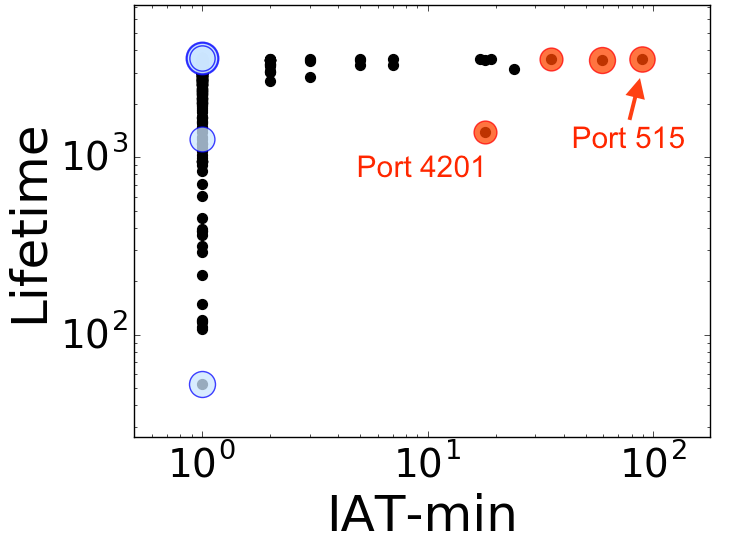}\\
		(c) & (d)\\
		\includegraphics[width=\mywidth\columnwidth]{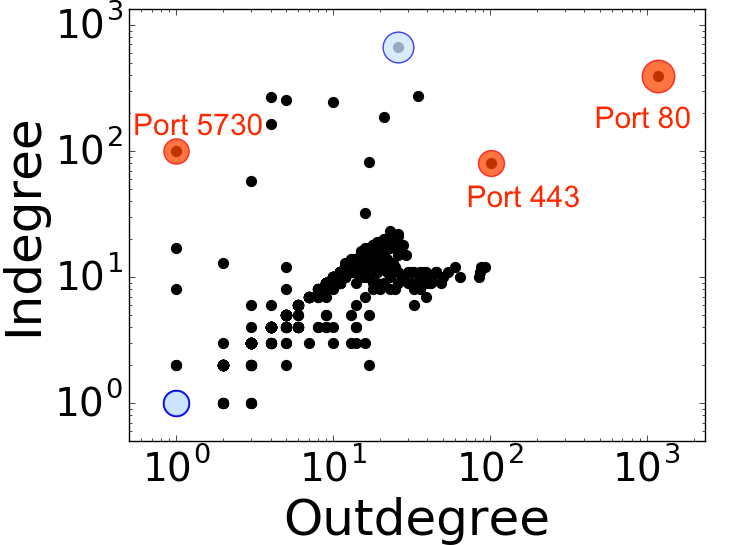}&
		\includegraphics[width=\mywidth\columnwidth]{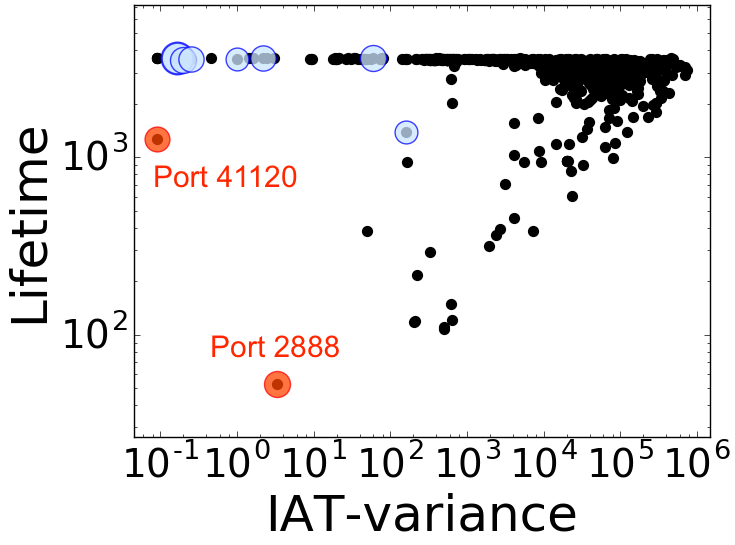}\\
		(e) & (f)
	\end{tabular}
	\caption{\textbf{Discoveries using \method on detected anomalies:} \method partitions and explains anomaly detection results from iForest on \dblp (a-b) and \lbnl (c-f) datasets. \label{fig:discovery:detected}}
\end{figure}

In this section, we present our discoveries using \method on all three real world datasets. Scoring in 2-d was performed using iForest with $t = 100$ trees and sample size $\psi=64$ (\enron) and $\psi=256$ (\lbnl and \dblp).  We use dictated anomalies for \enron, and detected anomalies for \lbnl and \dblp datasets to demonstrate performance in both settings.

\vspace{1mm}
\noindent \enron: We used two top actors in the \enron scandal, \texttt{Kenneth Lay} (CEO) and \texttt{Jeff Skilling} (CFO) as dictated anomalies for \method and sought explanations for their abnormality based on internal e-mail communications. With $b$=2, \method produced the plots shown in Fig.~\ref{fig:jewel} (right). Explanations indicate that \texttt{Jeff Skilling} had an unsually large \iatmax for the number of employees he communicated with (\outd). On the other hand, \texttt{Kenneth Lay} sent emails to an abnormally large number of employees (\outd) given the time range during which he emailed anyone (\lifetime).

\vspace{1mm}
\noindent \dblp: We obtained ground truth anomalies by running iForest on the high-dimensional space spanned by a subset of features from Sec.~\ref{ssec:features}. With $k=5$, the detected anomalous authors were \texttt{Alberto L. Sangiovanni-Vincentelli}, \texttt{H. Vincent Poor}, \texttt{Hao Wang}, \texttt{Jack Dongarra} and \texttt{Thomas S. Huang}. The explanations provided by \method with $b=2$ are shown in Fig.~\ref{fig:discovery:detected}a-b. Thus, the anomalous authors users are partitioned into two groups. The members of the first group, \texttt{Alberto L. Sangiovanni-Vincentelli}, \texttt{H. Vincent Poor}, \texttt{Jack Dongarra}, and \texttt{Thomas S. Huang} are anomalous because they had usually high duration during which they published papers (\lifetime) and total number of co-authorships (\inwr). This is consistent with their high h-indices obtained from their respective google scholar pages (see brackets in Fig.~\ref{fig:discovery:detected}a). The second group consists of only \texttt{Hao Wang}, who was also anomalous in this plot, but is best explained by very high \iatvar for his \inwr value, shown in Fig.~\ref{fig:discovery:detected}b.

\vspace{1mm}
\noindent \lbnl: We ran \method with $k$=10 and $b$=4 and obtained the output pair plots as shown in Fig.~\ref{fig:discovery:detected}c-f. We observe the different anomalous nodes that have been incriminated differently depending on the pair plot. \texttt{Port 524} is highlighted in Fig.~\ref{fig:discovery:detected}c due to its very low \iatavg and \iatvar hinting at it being hotspot port with a large number of transmissions at a quick continuous rate. 

Fig.~\ref{fig:discovery:detected}d shows 4 anomalies including \texttt{Port 515} and \texttt{Port 4201}. In fact, \texttt{Port 515} handles requests based on the Line Printer Daemon Protocol which establishes network printing services. Our observations of a high \iatmin over a long \lifetime are in tandem with print requests being polled over the network at fixed time intervals.

Fig.~\ref{fig:discovery:detected}e incriminates \texttt{Port 80}, \texttt{Port 443} and \texttt{Port 5730}. \texttt{Port 80} is the designated port for the internet hypertext transfer protocol (HTTP) and helps retrieve HTML data from web servers. As such, we observe that \texttt{Port 80} has an extremely high \ind and \outd, indicating that this port is used for a large number of diverse connections between several web servers. Its contemporary counterpart \texttt{Port 443} is assigned for secure communications (HTTPS) and is thus incriminated for the same identifying reasons as the former. 

\texttt{Port 2888} in Fig.~\ref{fig:discovery:detected}f is used by ZooKeeper, a centralized coordination service for distributed applications like \hadoop\footnote{\url{https://wiki.apache.org/hadoop/ZooKeeper}}. It appears with an abnormally low \lifetime and \iatvar probably indicating a small active transmission time which involved a quick configuration broadcast between devices.

\vspace{1mm}

Note that on all datasets, anomalous points are clearly visually distinguishable, and often complementary between pair plots.  This is in line with our desired explanation task, and achieved as a result of our \method subset selection objective and approach.

%% file: 020background.tex
While there is considerable prior work on anomaly detection \cite{chandola2009anomaly,journals/datamine/AkogluTK15,journals/tkde/GuptaGAH14}, literature on anomaly description is comparably sparse.  Several works aim to find an optimal feature subspace which distinguishes anomalies from normal points.  \cite{conf/cikm/KellerMWB13} aims to find a subspace which maximizes differences in anomaly score distributions of all points across subspaces.  \cite{conf/aaai/KuoD16} instead takes a constraint programming approach which aims to maximize differences between neighborhood densities of known anomalies and normal points.  An associated problem focuses on finding minimal, or optimal feature subspaces for each anomaly.  \cite{conf/vldb/KnorrN99} aims to give ``intensional knowledge'' for each anomaly by finding minimal subspaces in which the anomalies deviate sufficiently from normal points using pruning rules.  \cite{conf/icde/DangANZS14,conf/pkdd/DangMAN13} use spectral embeddings to discover subspaces which promote high anomaly scores, while aiming to preserve distances of normal points.  \cite{conf/icdm/MicenkovaNDA13} instead employs sparse classification of an inlier class against a synthetically-created anomaly class for each anomaly in order to discover small feature spaces which discern it.  \cite{kopp2014interpreting} proposes combining decision rules produced by an ensemble of short decision trees to explain anomalies. \cite{journals/tkde/AngiulliFP13} augments the per-outlier problem to include outlier groups by searching for single features which differentiate many outliers.

All in all, none of these works meet several key desiderata for anomaly description: (a) quantifiable explanation quality, (b) budget-consciousness towards analysts (returning explanations which do not grow with size of the anomalous set), (c) visual interpretability, and (d) a scalable descriptor, which is sub-quadratic on the number of nodes and at worst polynomial on (low) dimensionality.  Table \ref{tbl:related_comparison} shows that unlike existing approaches, our \method approach is designed to give quantifiable explanations which aim to maximize \emph{incrimination}, respect human attention-budget and visual interpretability constraints, and scale linearly on the number of graph edges.

\begin{table}[!t]
	\caption{\label{tbl:related_comparison} Comparison with other anomaly description approaches, in terms of four desirable properties.}
	\centering
	\resizebox{0.49\textwidth}{!}{
		\begin{tabular}{l | llllll || l}
			\hline
			Properties vs. Methods                        & \rotatebox[origin=c]{90}{Knorr et al. \cite{conf/vldb/KnorrN99}} & \rotatebox[origin=c]{90}{Dang et al.\cite{conf/icde/DangANZS14}} & \rotatebox[origin=c]{90}{Angiulli et al. \cite{journals/tkde/AngiulliFP13}} & \rotatebox[origin=c]{90}{\, Micenkova et al. \cite{conf/icdm/MicenkovaNDA13}} & \rotatebox[origin=c]{90}{Kopp et al. \cite{kopp2014interpreting}} & \rotatebox[origin=c]{90}{Keller et al. \cite{conf/cikm/KellerMWB13}} & \rotatebox[origin=c]{90}{\method} \\
			\hline
			Quantifiable Explanations &                                  & \tick                                  & \tick                                            & \tick                                            &                    & \tick                     & \tick   \\
			Budget-conscious          &                                  &                                  &                                             &                                            &                 & \tick                        & \tick   \\
			Visually Interpretable    &                                  &                                  &                                             &                                            &               &                           & \tick   \\
			Scalable                  &                                  & \tick                                  &                                             &                                            & \tick                  & \tick                       & \tick   \\
			\hline
	\end{tabular}}
\end{table}

%% file: 060conclusion.tex
In this work, we formulated and tackled the problem of succinctly and interpretably explaining anomalies in $t$-graphs to human analysts.  We identified a number of domain-agnostic features of $t$-graphs, and subsequently made the following contributions: (a) \textbf{problem formulation}: we formulate our goal for explaining anomalies using a budget of visually interpretable pair plots, (b) \textbf{explanation algorithm:} we propose a submodular objective to quantify explanation quality and propose the \method method for solving it approximately with guarantees, (c) \textbf{generality}: we show that \method can work with diverse domains and any detection algorithm, and (d) \textbf{scalability:} we show theoretically and empirically that \method scales linearly in the size of the input graph.  We conduct experiments on real e-mail communication, co-authorship and IP-IP interaction $t$-graphs and demonstrate that \method produces qualitatively interpretable explanations for ``ground-truth'' anomalies and achieves strong quantitative performance in maximizing our proposed objective.


%% file: 070proof.tex
\section{Proof of Submodularity}
\label{proof}

Our \textbf{explanation objective} is defined as: 
$$
{\max_{{\mS\subset \mP, |\mS|=b}}}\;\;\; f(\mathcal{S}) = \sideset{}{\limits_{a_i \in \mathcal{A}}}\sum \;\;\max_{p_j \in \mathcal{S}}\;s_{i,j}  
$$
Next we show that the set function $f$ is submodular.

\begin{proof}
Consider two sets $\mathcal{S}$ and $\mathcal{T}$, $\mathcal{S} \subseteq \mathcal{T} \subseteq \mathcal{P}$, and
the function, $$\mathcal{F} (p_k, \mathcal{S}, \mathcal{T}) = \Delta_f(p_k \ | \ \mathcal{S}) - \Delta_f(p_k \ | \ \mathcal{T}) \;.$$

By definition of submodularity, function $f$ is submodular if and only if, 
$$\mathcal{F} (p_k, \mathcal{S}, \mathcal{T}) \geq 0 \;\; \forall \ p_k \in \mathcal{P} \setminus \mathcal{T}\;.$$


For simplicity lets write $\mathcal{F}$ as,
$$\mathcal{F} (p_k, \mathcal{S}, \mathcal{T}) =  \sideset{}{_{a_i \in \mathcal{A}}}\sum t_i$$
where each term $t_i$ is given as
\begin{multline*}
$$t_i = [\max(s_{i,k}, \max\limits_{p_j \in \mathcal{S}} s_{i,j}) \ -  \max(s_{i,k}, \max\limits_{p_j \in \mathcal{T}}  s_{i,j}  \\ + \max\limits_{p_j \in \mathcal{T}}\; s_{i,j} - \max\limits_{p_j \in \mathcal{S}}\; s_{i,j}]\;.$$
\end{multline*}	
Now we consider the following three cases:
\begin{case}
For some $a_i \in \mathcal{A} \ \exists \ p_{j'} \in \mathcal{S}$ s.t. $s_{i,j} \geq s_{i,k}$, then,
\vspace{-2mm}
$$\max(s_{i,k}, \max\limits_{p_j \in \mathcal{S}} s_{i,j}) = \max\limits_{p_j \in \mathcal{S}} s_{i,j}\;,$$
and similarly since $\mathcal{S} \subseteq \mathcal{T}, \ p_{j'} \in \mathcal{T}$
$$\max(s_{i,k}, \max\limits_{p_j \in \mathcal{T}} s_{i,j}) = \max\limits_{p_j \in \mathcal{T}} s_{i,j}\;.$$
$\therefore$ \ for all such $a_i \in \mathcal{A}, \ t_i = 0\;.$
\end{case}

\begin{case}
For some $a_i \in \mathcal{A} \ \nexists \ p_{j'} \in \mathcal{S}$ s.t. $s_{i,j'} \geq s_{i,k}$ and $\exists \ p_{j'} \in \mathcal{T}$ s.t. $s_{i,j'} \geq s_{i,k}$ , then,
$$\max(s_{i,k}, \max\limits_{p_j \in \mathcal{S}} s_{i,j}) = s_{i,k}$$
and conversely,
$$\max(s_{i,k}, \max\limits_{p_j \in \mathcal{T}} s_{i,j}) = \max\limits_{p_j \in \mathcal{T}} s_{i,j}$$
$\therefore$ \ for all such $a_i \in \mathcal{A},$ 
$t_i = s_{i,k} - \max\limits_{p_j \in \mathcal{S}} s_{i,j} \geq 0\;.$
\end{case}

\begin{case}
For some $a_i \in \mathcal{A} \ \nexists \ p_{j'} \in \mathcal{T}$ s.t. $s_{i,j'} > s_{i,k}$, then,
\vspace{-2mm}
$$\max(s_{i,k}, \max\limits_{p_j \in \mathcal{T}} s_{i,j}) = s_{i,k}$$
and similarly since $\mathcal{S} \subseteq \mathcal{T}$,
$$\max(s_{i,k}, \max\limits_{p_j \in \mathcal{S}} s_{i,j}) = s_{i,k}\;.$$
$\therefore$ \ for all such $a_i \in \mathcal{A}$, 
\vspace{-1mm}
\begin{multline*}
$$t_i = \max\limits_{p_j \in \mathcal{T}} s_{i,j} -
\max\limits_{p_j \in \mathcal{S}} s_{i,j}
 \geq 0 \ \ (\because \mathcal{A} \subseteq \mathcal{B})$$
\end{multline*}
\end{case}
\vspace{2mm}

As a result of all three possible cases we see that $t_i \geq 0 \ \forall \ a_i \in \mathcal{A}$. Hence,
$$\mathcal{F} (p_k, \mathcal{S}, \mathcal{T}) =  \sideset{}{_{a_i \in \mathcal{A}}}\sum t_i \geq 0$$
$$\therefore \ \ \mathcal{F} (p_k, \mathcal{S}, \mathcal{T}) \geq 0 \ \forall \ p_k \in \mathcal{P} \setminus \mathcal{T}$$
We conclude that the set function $f$ is submodular. \hfill $\blacksquare$
\end{proof}